\documentclass{llncs}

\usepackage{oldlfont,amssymb,epsf,stmaryrd}

\usepackage{color}

\newcommand\VL[1]{#1}
\newcommand\VC[1]{}

\newbox\tempa
\newbox\tempb
\newdimen\tempc
\def\mud#1{\hfil $\displaystyle{\mathstrut #1}$\hfil}
\def\rig#1{\hfil $\displaystyle{#1}$}
\def\irulehelp#1#2#3{\setbox\tempa=\hbox{$\displaystyle{\mathstrut #2}$}%
                        \setbox\tempb=\vbox{\halign{##\cr
        \mud{#1}\cr
        \noalign{\vskip\the\lineskip}%
        \noalign{\hrule height 0pt}%
        \rig{\vbox to 0pt{\vss\hbox to 0pt{${\; #3}$\hss}\vss}}\cr
        \noalign{\hrule}%
        \noalign{\vskip\the\lineskip}%
        \mud{\copy\tempa}\cr}}%
                      \tempc=\wd\tempb
                      \advance\tempc by \wd\tempa
                      \divide\tempc by 2 }
\def\irule#1#2#3{{\irulehelp{#1}{#2}{#3}%
                     \hbox to \wd\tempa{\hss \box\tempb \hss}}}

\def\fa{\forall}
\def\Ra{\Rightarrow}
\def\lra{\longrightarrow}

\def\apa{\mbox{ $\lra$ \hspace{-1.70em} $\shortmid$ \hspace{-.70em} $\shortmid$ \hspace{0.7em}}}
\def\apar{\mbox{ $\lra^{*}$ \hspace{-2.15em} $\shortmid$ \hspace{-.70em} $\shortmid$ \hspace{1.15em}}}

\begin{document}
\title{Embedding Pure Type Systems\\ in the lambda-Pi-calculus modulo}
\author{Denis Cousineau and Gilles Dowek}
\institute{\'Ecole polytechnique and INRIA\\
LIX, \'Ecole polytechnique,
91128 Palaiseau Cedex, France. \\ ~ \\
{\scriptsize
{\tt Cousineau@lix.polytechnique.fr, http://www.lix.polytechnique.fr/\~{}cousineau} \\
{\tt Gilles.Dowek@polytechnique.edu, http://www.lix.polytechnique.fr/\~{}dowek}
}
}
\maketitle

\begin{abstract}
The lambda-Pi-calculus allows
to express proofs of minimal predicate logic. It can be
extended, in a very simple way, by adding computation rules. This leads
to the lambda-Pi-calculus modulo.  We show in this paper that this
simple extension is surprisingly expressive and, in particular, that
all functional Pure Type Systems, such as the system F, or the
Calculus of Constructions, can be embedded in it. And, moreover, 
that this embedding is conservative under termination hypothesis.
\end{abstract}

The {\em $\lambda \Pi$-calculus} is a dependently typed
lambda-calculus that allows to express proofs of minimal predicate
logic through the Brouwer-Heyting-Kolmogorov interpretation and the
Curry-de Bruijn-Howard correspondence.  It can be extended in several
ways to express proofs of some theory. A first solution is to express
the theory in Deduction modulo \cite{DHK,DW}, {\em i.e.}  to orient
the axioms as rewrite rules and to extend the $\lambda \Pi$-calculus
to express proofs in Deduction modulo \cite{Blanqui}. We get this way
the {\em $\lambda \Pi$-calculus modulo}.  This idea of extending the
dependently typed lambda-calculus with rewrite rules is also that of
Intuitionistic type theory used as a logical framework \cite{NPS}.

A second way to extend the $\lambda \Pi$-calculus is to add typing
rules, in particular to allow polymorphic typing. We get this way the
{\em Calculus of Constructions} \cite{CoquandHuet} that allows to
express proofs of simple type theory and more generally the {\em Pure
Type Systems} \cite{Berardi,Terlouw,Barendregt}. These two kinds of
extensions of the $\lambda \Pi$-calculus are somewhat redundant.  For
instance, simple type theory can be expressed in deduction modulo
\cite{DHKHOL}, hence the proofs of this theory can be expressed in the
$\lambda \Pi$-calculus modulo. But they can also be expressed in the
Calculus of Constructions. This suggests to relate and compare these
two ways to extend the $\lambda \Pi$-calculus.

We show in this paper that all functional Pure Type Systems can be
embedded in the $\lambda \Pi$-calculus modulo using an appropriate
rewrite system.  This rewrite system is inspired both by the
expression of simple type theory in Deduction modulo and
by the mechanisms of universes {\em \`a la} Tarski \cite{MartinLof} of
Intuitionistic type theory. In particular, this work extends
Palmgren's construction of an impredicative universe in type theory
\cite{Palmgren}.

\section{The $\lambda \Pi$-calculus}

The {\em $\lambda \Pi$-calculus} is a dependently typed
lambda-calculus that permits to construct types depending on terms,
for instance a type $array~n$, of arrays of size $n$, that depends on
a term $n$ of type $nat$. It also permits to construct a function $f$
taking a natural number $n$ as an argument and returning an array of
size $n$. Thus, the arrow type $nat \Ra array$ of simply typed
lambda-calculus must be extended to a dependent product type $\Pi
x:nat~(array~x)$ where, in the expression $\Pi x:A~B$, the occurrences
of the variable $x$ are bound in $B$ by the symbol $\Pi$ (the
expression $A \Ra B$ is used as a shorter notation for the expression
$\Pi x:A~B$ when $x$ has no free occurrence in $B$).  When we apply
the function $f$ to a term $n$, we do not get a term of type
$array~x$ but of type $array~n$. Thus, the application rule must
include a substitution of the term $n$ for the variable $x$.  The symbol
$array$ itself takes a natural number as an argument and returns a
type. Thus, its type is $nat \Ra Type$, {\em i.e.} $\Pi x:nat~Type$.
The terms $Type$, $nat \Ra Type$, ... cannot have type $Type$, because
Girard's paradox \cite{Girard} could then be expressed in the system,
thus we introduce a new symbol $Kind$ to type such terms.  To form
terms, like $\Pi x:nat~Type$, whose type is $Kind$, we need a rule
expressing that the symbol $Type$ has type $Kind$ and a new product
rule allowing to form the type $\Pi x:nat~Type$, whose type is
$Kind$. Besides the variables such as $x$ whose type has type $Type$, we must permit the declaration of variables such as $nat$ of type $Type$, and more generally, variables such as $array$ whose
type has type $Kind$.  This leads to introduce the following syntax
and typing rules.

\begin{definition}[The syntax of $\lambda \Pi$]
The syntax of the $\lambda \Pi$-calculus is 
$$t = x~|~Type~|~Kind~|~\Pi x:t~t~|~\lambda x:t~t~|~t~t$$
The $\alpha$-equivalence and $\beta$-reduction relations are defined as
usual and terms are identified modulo $\alpha$-equivalence. 
\end{definition}

\begin{definition}[The typing rules of $\lambda \Pi^{-}$]
$$\irule{} 
        {[~]~\mbox{well-formed}}
        {\mbox{\bf Empty}}$$
$$\irule{\Gamma \vdash A:Type} 
        {\Gamma[x:A]~~\mbox{well-formed}}
        {\mbox{{\bf Declaration}~$x$ not in $\Gamma$}}$$
$$\irule{\Gamma \vdash A:Kind} 
        {\Gamma[x:A]~~\mbox{well-formed}}
        {\mbox{{\bf Declaration2}~$x$ not in $\Gamma$}}$$
$$\irule{\Gamma~\mbox{well-formed}} 
        {\Gamma \vdash Type:Kind}
        {\mbox{\bf Sort}}$$
$$\irule{\Gamma~\mbox{well-formed}~~x:A \in \Gamma} 
        {\Gamma \vdash x:A}
        {\mbox{\bf Variable}}$$
$$\irule{\Gamma \vdash A:Type~~\Gamma[x:A] \vdash B:Type}
        {\Gamma \vdash \Pi x:A~B:Type}
        {\mbox{\bf Product}}$$
$$\irule{\Gamma \vdash A:Type~~\Gamma[x:A] \vdash B:Kind}
        {\Gamma \vdash \Pi x:A~B:Kind}
        {\mbox{\bf Product2}}$$
$$\irule{\Gamma \vdash A:Type~~\Gamma[x:A] \vdash B:Type~~\Gamma[x:A] \vdash t:B} 
        {\Gamma \vdash \lambda x:A~t:\Pi x:A~B} 
        {\mbox{\bf Abstraction}}$$ 
$$\irule{\Gamma \vdash t:\Pi x:A~B~~\Gamma \vdash u:A} 
        {\Gamma \vdash (t~u):(u/x)B} 
        {\mbox{\bf Application}}$$ 
\end{definition}

It is useful, in some situations, to add a rule allowing
to build type families by abstraction, for instance 
$\lambda x:nat~(array~(2 \times x))$
and rules asserting that a term of type 
$(\lambda x:nat~(array~(2 \times x))~n)$ also has type $array~(2 \times n)$.
This leads to introduce the following extra typing rules.

\begin{definition}[The typing rules of $\lambda \Pi$]
The typing rules of $\lambda \Pi$ are those of $\lambda \Pi^{-}$ and
$$\irule{\Gamma \vdash A:Type~~\Gamma[x:A] \vdash B:Kind~~\Gamma[x:A]
        \vdash t:B} {\Gamma \vdash \lambda x:A~t:\Pi x:A~B}
        {\mbox{\bf Abstraction2}}$$
$$\irule{\Gamma \vdash A:Type~~\Gamma \vdash B:Type~~\Gamma \vdash t:A}
        {\Gamma \vdash t:B}
        {\mbox{\bf Conversion}~A \equiv_{\beta} B}$$
$$\irule{\Gamma \vdash A:Kind~~\Gamma \vdash B:Kind~~\Gamma \vdash t:A}
        {\Gamma \vdash t:B}
        {\mbox{\bf Conversion2}~A \equiv_{\beta} B}$$
where $\equiv_{\beta}$ is the $\beta$-equivalence relation.
\end{definition}

It can be proved that types are preserved by $\beta$-reduction, that
$\beta$-reduction is confluent and strongly terminating and that each
term has a unique type modulo $\beta$-equivalence.

The $\lambda \Pi$-calculus, and even the $\lambda \Pi^-$-calculus, 
can be used to express proofs of minimal predicate
logic, following the Brouwer-Heyting-Kolmogorov interpretation and 
the Curry-de Bruijn-Howard correspondence. 
Let ${\cal L}$ be a language in predicate logic, we consider a context
$\Gamma$ formed with 
a variable $\iota$ of type $Type$ -- or variables $\iota_1, ...,
\iota_n$ of type $Type$ when ${\cal L}$ is many-sorted ---, 
for each function symbol $f$ of ${\cal L}$, a variable $f$ of
type $\iota \Ra ... \Ra \iota \Ra \iota$ and
for each predicate symbol $P$ of ${\cal L}$, a variable $P$ of
type $\iota \Ra ... \Ra \iota \Ra Type$.

To each formula $P$ containing free variables $x_1, ..., x_p$ we
associate a term $P^{\circ}$ of type $Type$ in the context $\Gamma,
x_1:\iota, ..., x_p:\iota$ translating each variable, function symbol
and predicate symbol by itself and the implication symbol and the
universal quantifier by a product.

To each proof $\pi$, in minimal natural deduction, of a sequent
$A_1, ..., A_n \vdash B$  with free variables $x_1, ..., x_p$, 
we can associate a term $\pi^{\circ}$ of type $B^{\circ}$ in the
context $\Gamma, x_1:\iota, ..., x_p:\iota, 
\alpha_1:A_1^{\circ}, ..., \alpha_n:A_n^{\circ}$.
{From} the strong termination of the $\lambda \Pi$-calculus, we get cut
elimination for minimal predicate logic. If $B$ is an atomic formula,
there is no cut free proof, hence no proof at all, of $\vdash B$.

\section{The $\lambda \Pi$-calculus modulo}

The $\lambda \Pi$-calculus allows to express proofs in pure minimal
predicate logic. To express proofs in a theory ${\cal T}$, we can
declare a variable for each axiom of ${\cal T}$ and consider
proofs-terms containing such free variables, this is the idea of the
Logical Framework \cite{HHP}.
However, when considering such open terms most benefits of
termination, such as the existence of empty types,
are lost.  

An alternative is to replace axioms by rewrite rules, moving from
predicate logic to {\em Deduction modulo} \cite{DHK,DW}.  Such
extensions of type systems with rewrite rules to express proofs in
Deduction modulo have been defined in \cite{Blanqui} and \cite{NPS}.  
We shall present now an extension of the $\lambda
\Pi$-calculus: the $\lambda \Pi$-calculus modulo.

Recall that if $\Sigma$, $\Gamma$ and $\Delta$ are contexts, a
substitution $\theta$, binding the variables declared in $\Gamma$, is
said to be {\em of type $\Gamma \leadsto \Delta$ 
in $\Sigma$} if for all $x$ declared of type $T$ in $\Gamma$, we have
$\Sigma \Delta \vdash \theta x:\theta T$,
and that, in this case, if $\Sigma \Gamma \vdash u:U$, then $\Sigma
\Delta \vdash \theta u:\theta U$.

A {\em rewrite rule} is a quadruple $l \lra^{\Gamma, T} r$ where $\Gamma$ is
a context and $l$, $r$ and $T$ are $\beta$-normal terms.
Such a rule is
said to be {\em well-typed} in the context $\Sigma$ if, in the $\lambda
\Pi$-calculus, the context $\Sigma \Gamma$ is well-formed and the
terms $l$ and $r$ have type $T$ in this context.

If $\Sigma$ is a context, $l \lra^{\Gamma,T} r$ is a rewrite
rule well-typed in $\Sigma$ and $\theta$ is a substitution of type
$\Gamma \leadsto \Delta$ in $\Sigma$ then the terms $\theta l$ and
$\theta r$ both have type $\theta T$ in the context $\Sigma \Delta$.
We say that the term $\theta l$ {\em rewrites} to the term $\theta r$.

If $\Sigma$ is a context and ${\cal R}$ a set of rewrite rules
well-typed in the $\lambda \Pi$-calculus in $\Sigma$, then the {\em
congruence generated by ${\cal R}$}, $\equiv_{\cal R}$, is the smallest
congruence such that if $t$ rewrites to $u$ then $t \equiv_{\cal R}
u$.

\begin{definition}[$\lambda \Pi$-modulo]
Let $\Sigma$ be a context and ${\cal R}$ a rewrite system 
in $\Sigma$. Let $\equiv_{\beta {\cal R}}$ be the congruence of terms
generated by the rules of ${\cal R}$ and the rule $\beta$. 

The {\em $\lambda \Pi$-calculus modulo ${\cal R}$} is the 
extension of the $\lambda \Pi$-calculus
obtained by replacing the relation $\equiv_{\beta}$ by $\equiv_{\beta
  {\cal R}}$
in the conversion rules
$$\irule{\Gamma \vdash A:Type~~\Gamma \vdash B:Type~~\Gamma \vdash t:A}
        {\Gamma \vdash t:B}
        {\mbox{\bf Conversion}~A \equiv_{\beta {\cal R}} B}$$
$$\irule{\Gamma \vdash A:Kind~~\Gamma \vdash B:Kind~~\Gamma \vdash t:A}
        {\Gamma \vdash t:B}
        {\mbox{\bf Conversion2}~A \equiv_{\beta {\cal R}} B}$$
\end{definition}

Notice that we can also extend the $\lambda \Pi^-$-calculus with rewrite
rules. In this case, we introduce conversion rules, using the 
congruence defined by the system ${\cal R}$ alone.

\begin{example}
Consider the context $\Sigma = [P:Type,~ Q:Type]$ and the rewrite
system ${\cal R}$ formed with the rule $P \lra (Q \Ra Q)$. The term
$\lambda f:P~\lambda x:Q~(f~x)$ is well-typed in the $\lambda
\Pi$-calculus modulo ${\cal R}$. 
\end{example}

\section{The Pure Type Systems}

There are several ways to extend the functional interpretation of
proofs to simple type theory. The first is to use the fact that simple
type theory can be expressed in Deduction modulo with rewrite rules
only \cite{DHKHOL}. Thus, the proofs of simple type theory can be
expressed in the $\lambda \Pi$-calculus modulo, and even in the
$\lambda \Pi^-$-calculus modulo.  The second is to extend the $\lambda
\Pi$-calculus by adding the following typing rules, allowing for instance the
construction of the type $\Pi P:Type~(P \Rightarrow P)$.
$$\irule{\Gamma \vdash A:Kind~~\Gamma[x:A] \vdash B:Type}
        {\Gamma \vdash \Pi x:A~B:Type}
        {\mbox{\bf Product3}}$$
$$\irule{\Gamma \vdash A:Kind~~\Gamma[x:A] \vdash B:Kind}
        {\Gamma \vdash \Pi x:A~B:Kind}
        {\mbox{\bf Product4}}$$
$$\irule{\Gamma \vdash A:Kind~~\Gamma[x:A] \vdash B:Type~~\Gamma[x:A] \vdash t:B} 
        {\Gamma \vdash \lambda x:A~t:\Pi x:A~B}
        {\mbox{\bf Abstraction3}}$$
$$\irule{\Gamma \vdash A:Kind~~\Gamma[x:A] \vdash B:Kind~~\Gamma[x:A] \vdash t:B} 
        {\Gamma \vdash \lambda x:A~t:\Pi x:A~B}
        {\mbox{\bf Abstraction4}}$$
We obtain the {\em Calculus of Constructions} \cite{CoquandHuet}.

The rules of the simply typed $\lambda$-calculus, the $\lambda
\Pi$-calculus and of the Calculus of Constructions can be presented in
a schematic way as follows.

\begin{definition}[Pure type system]
A {\em Pure Type System} \cite{Berardi,Terlouw,Barendregt} $P$ is
defined by a set $S$, whose elements are 
called {\em sorts}, a subset $A$ of $S \times S$, whose elements are
called {\em axioms} and a subset $R$ of $S \times S \times S$, whose
elements are called {\em rules}. The typing rules of $P$ are 
$$\irule{} 
        {[~]~\mbox{well-formed}}
        {\mbox{\bf Empty}}$$
$$\irule{\Gamma \vdash A:s} 
        {\Gamma[x:A]~~\mbox{well-formed}}
        {\mbox{{\bf Declaration}~$s \in S$ and $x$ not in $\Gamma$}}$$
$$\irule{\Gamma~\mbox{well-formed}} 
        {\Gamma \vdash s_1:s_2}
        {\mbox{\bf Sort}~\langle s_1,s_2\rangle  \in A}$$
$$\irule{\Gamma~\mbox{well-formed}~~x:A \in \Gamma} 
        {\Gamma \vdash x:A}
        {\mbox{\bf Variable}}$$
$$\irule{\Gamma \vdash A:s_1~~\Gamma[x:A] \vdash B:s_2}
        {\Gamma \vdash \Pi x:A~B:s_3}
        {\mbox{\bf Product}~\langle s_1,s_2,s_3\rangle  \in R}$$
$$\irule{\Gamma \vdash A:s_1~~\Gamma[x:A] \vdash B:s_2~~\Gamma[x:A] \vdash t:B} 
        {\Gamma \vdash \lambda x:A~t:\Pi x:A~B}
        {\mbox{\bf Abstraction}~\langle s_1, s_2, s_3 \rangle \in R}$$
$$\irule{\Gamma \vdash t:\Pi x:A~B~~\Gamma \vdash u:A}
        {\Gamma \vdash (t~u):(u/x)B}
        {\mbox{\bf Application}}$$
$$\irule{\Gamma \vdash A:s~~\Gamma \vdash B:s~~\Gamma \vdash t:A}
        {\Gamma \vdash t:B}
        {\mbox{\bf Conversion}~s \in S~~A \equiv_{\beta} B}$$
\end{definition}

The simply typed $\lambda$-calculus is the system defined by the sorts
$Type$ and $Kind$, the axiom $\langle Type,Kind \rangle$ and the rule
$\langle Type, Type, Type \rangle$. 
The $\lambda \Pi$-calculus is the system defined by
the same sorts and axiom and the rules 
$\langle Type, Type, Type \rangle$ and $\langle Type, Kind, Kind
\rangle$. The Calculus of Constructions is the system defined by 
the same sorts and axiom and the rules $\langle Type, Type, Type
\rangle$, $\langle Type, Kind, Kind \rangle$, $\langle Kind, Type,
Type \rangle$ and $\langle Kind, Kind, Kind \rangle$. Other
examples of Pure Type Systems are Girard's systems F and F$\omega$.

In all Pure Type Systems, types are preserved under reduction and the
$\beta$-reduction relation is confluent.  It terminates in some
systems, such as the $\lambda \Pi$-calculus, the Calculus of
Constructions, the system F and the system F$\omega$. Uniqueness of types
is lost in general, but it holds for the $\lambda \Pi$-calculus, the
Calculus of Constructions, the system F and the system F$\omega$, and
more generally for all functional Pure Type Systems.

\begin{definition}[Functional Type System]
A type system is said to be {\em functional} if
$$\langle s_1, s_2 \rangle \in A~\mbox{and}~\langle s_1 ,s_3 \rangle
\in A~\mbox{implies}~s_2 = s_3$$
$$\langle s_1,s_2,s_3\rangle  \in R~\mbox{and}~\langle s_1,s_2,s_4\rangle  \in R~\mbox{implies}~s_3 = s_4$$
\end{definition}

\section{Embedding functional Pure Type Systems in the $\lambda
\Pi$-calculus modulo} 

We have seen that the $\lambda \Pi$-calculus modulo and the Pure Type
Systems are two extensions of the $\lambda \Pi$-calculus. At a first
glance, they seem quite different as the latter adds more typing rules
to the $\lambda \Pi$-calculus, while the former adds more computation
rules. But they both allow to express proofs of simple type theory.

We show in this section that functional Pure Type Systems can, in fact, be
embedded in the $\lambda \Pi$-calculus modulo with an appropriate 
rewrite system.

\subsection{Definition}

Consider a functional Pure Type System $P = \langle S, A, R \rangle$.
We build the following context and rewrite system.

The context $\Sigma_{P}$ contains, for each sort $s$, two variables
$$U_s:Type~~~~and~~~~ \varepsilon_s:U_s \Ra Type$$ 
for each axiom $\langle s_1, s_2 \rangle$, a variable
$$\dot{s_1}:U_{s_2}$$  
and for each rule $\langle s_1, s_2, s_3 \rangle$, a variable
$$\dot{\Pi}_{\langle s_1, s_2, s_3 \rangle}:
\Pi X:U_{s_1}~(((\varepsilon_{s_1}~X) \Ra U_{s_2}) \Ra U_{s_3})$$

The type $U_s$ is called the {\em universe} of $s$ and the symbol
$\varepsilon_{s}$ the {\em decoding function} of $s$.

The rewrite rules are 
$$\varepsilon_{s_2} (\dot{s}_1) \lra U_{s_1}$$
in the empty context and with the type $Type$, and 
$$\varepsilon_{s_3} (\dot{\Pi}_{\langle s_1, s_2, s_3
  \rangle}~X~Y) \lra 
\Pi x:(\varepsilon_{s_1}~X)~(\varepsilon_{s_2}~(Y~x))$$
in the context $X:U_{s_{1}}, Y:(\varepsilon_{s_{1}}~X)
\Rightarrow U_{s_{2}}$ and with the type $Type$.

\medskip

These rules are called the {\em universe-reduction rules}, 
we write $\equiv_{P}$ for the
equivalence relation generated by these rules and the rule $\beta$
and we call the $\lambda \Pi_{P}$-calculus 
the $\lambda \Pi$-calculus modulo these rewrite rules and the
rule $\beta$. To ease
notations, in the $\lambda \Pi_{P}$-calculus, we do not recall the context
$\Sigma_{P}$ in each sequent and write $\Gamma \vdash
t:T$ for $\Sigma_{P} \Gamma \vdash t:T$, and we note $\equiv$ for $\equiv_P$ when there is no ambiguity.

\begin{example}
The embedding of the Calculus of Constructions is defined
by the context
$$\dot{Type}:U_{Kind}~~~~~~~~~~~~U_{Type}:Type~~~~~~~~~~~~U_{Kind}:Type$$
$$\varepsilon_{Type}:U_{Type} \Ra Type~~~~~~~~~~~~\varepsilon_{Kind}:U_{Kind} \Ra Type$$
$$\dot{\Pi}_{\langle Type, Type, Type \rangle} : \Pi X:U_{Type}~(((\varepsilon_{Type}~X) \Ra U_{Type}) \Ra U_{Type})$$
$$\dot{\Pi}_{\langle Type, Kind, Kind \rangle} : \Pi X:U_{Type}~(((\varepsilon_{Type}~X) \Ra U_{Kind}) \Ra U_{Kind})$$
$$\dot{\Pi}_{\langle Kind, Type, Type \rangle} : \Pi X:U_{Kind}~(((\varepsilon_{Kind}~X) \Ra U_{Type}) \Ra U_{Type})$$
$$\dot{\Pi}_{\langle Kind, Kind, Kind \rangle} : \Pi X:U_{Kind}~(((\varepsilon_{Kind}~X) \Ra U_{Kind}) \Ra U_{Kind})$$
and the rules 
$$\varepsilon_{Kind} (\dot{Type}) \lra U_{Type}$$
$$\varepsilon_{Type} (\dot{\Pi}_{\langle Type, Type, Type \rangle}~X~Y) \lra 
\Pi x:(\varepsilon_{Type}~X)~(\varepsilon_{Type}~(Y~x))$$
$$\varepsilon_{Kind} (\dot{\Pi}_{\langle Type, Kind, Kind \rangle}~X~Y) \lra 
\Pi x:(\varepsilon_{Type}~X)~(\varepsilon_{Kind}~(Y~x))$$
$$\varepsilon_{Type} (\dot{\Pi}_{\langle Kind, Type, Type \rangle}~X~Y) \lra 
\Pi x:(\varepsilon_{Kind}~X)~(\varepsilon_{Type}~(Y~x))$$
$$\varepsilon_{Kind} (\dot{\Pi}_{\langle Kind, Kind, Kind \rangle}~X~Y) \lra 
\Pi x:(\varepsilon_{Kind}~X)~(\varepsilon_{Kind}~(Y~x))$$
\end{example}

\begin{definition}[Translation]
Let $\Gamma$ be a context in a functional Pure Type System $P$ and $t$ a term
well-typed in $\Gamma$, we defined the translation $|t|$ of $t$ in
$\Gamma$, that is a term in $\lambda \Pi_{P}$, as follows
\begin{itemize}
\item $|x| = x$,
\item $|s| = \dot{s}$,
\item $|\Pi x:A~B| = \dot{\Pi}_{\langle s_1,s_2,s_3 \rangle}~|A|~(\lambda
x:(\varepsilon_{s_1}~|A|)~|B|)$, where $s_1$ is the type of $A$, $s_2$ is
the type of $B$ and $s_3$ the type of $\Pi x:A~B$, 
\item $|\lambda x:A~t| = \lambda x:(\varepsilon_{s}~|A|)~|t|$,
\item $|t~u| = |t|~|u|$.
\end{itemize}
\end{definition}

\begin{definition}[Translation as a type]
Consider a term $A$ of type $s$ for some sort $s$. The
{\em translation of $A$ as a type} is
$$\| A \| = \varepsilon_s~|A|.$$
Note that if $A$ is a well-typed sort $s'$ then 
$$\|s'\| = \varepsilon_s~\dot{s}' \equiv_{P} U_{s'}.$$
We extend this definition to non well-typed sorts,
such as the sort $Kind$ in the Calculus of Constructions, by
$$\|s'\| = U_{s'}$$
The translation of a well formed context is defined by
\begin{center}
 $\|[~]\| = [~]$ ~~and~~~ $ \|\Gamma [x:A]\| = \|\Gamma\|[x:\|A\|]$
\end{center}
\end{definition}

\subsection{Soundness}

\begin{proposition}
\label{reduct}
\begin{enumerate}
\item $|(u/x)t| = (|u|/x)|t|$, $\|(u/x)t\| = (|u|/x)\|t\|$. 
\item 
If $t \lra_{\beta} u$ then $|t| \lra_{\beta} |u|$. 
\end{enumerate}
\end{proposition}

\begin{proof}
\begin{enumerate}
\item By induction on $t$.
\item Because a $\beta$-redex is translated as a $\beta$-redex.
\end{enumerate}
\end{proof}

\begin{proposition}
\label{Piequiv}
$\|\Pi x:A~B\| \equiv_{P} \Pi x:\|A\|~\|B\|$
\end{proposition}

\begin{proof}
Let $s_1$ be the type of $A$, $s_2$ that of $B$ and $s_3$ that of $\Pi
x:A~B$. We have 
$\|\Pi x:A~B\| = \varepsilon_{s_3}~|\Pi x:A~B| =
\varepsilon_{s_3}~(\dot{\Pi}_{\langle s_{1}, s_{2}, s_{3}
  \rangle}~|A|~(\lambda x:(\varepsilon_{s_1}~|A|)~|B|)) \\ \equiv_{P}
\Pi x:(\varepsilon_{s_1}~|A|)~(
\varepsilon_{s_2}~((\lambda x:(\varepsilon_{s_1}~|A|)~|B|)~x)) \equiv_{P}
\Pi x:(\varepsilon_{s_1}~|A|)~(\varepsilon_{s_2}~|B|) \\ = 
\Pi x:\|A\|~\|B\|$.
\end{proof}

\begin{example}
In the Calculus of Constructions, the translation as a type of 
$\Pi X:Type~(X \Ra X)$ is
$\Pi X:U_{Type}~((\varepsilon_{Type}~X) \Ra (\varepsilon_{Type}~X))$. 
The translation as a term of $\lambda X:Type~\lambda x:X~x$ is the term
$\lambda X:U_{Type}~\lambda x:(\varepsilon_{Type} X)~x$. Notice that
the former is the type of the latter. The generalization of this
remark is the following proposition.
\end{example}

\begin{proposition}[Soundness]
\label{correction}
\begin{center}
If $~\Gamma \vdash t:B~$ in $P$
then $~\|\Gamma\|
\vdash |t|:\|B\|~$ in $\lambda \Pi_{P}$.
\end{center}
\end{proposition}

\begin{proof}
By induction on $t$. 
\begin{itemize}
\item If $t$ is a variable, this is trivial.
\item If $t = s_1$ then $B = s_2$ (where $\langle s_1,s_2 \rangle$ is an axiom), we have
$\dot{s_1}:U_{s_2} = \|s_2\|$.
\item If $t = \Pi x:C~D$, let $s_1$ be the type of $C$, $s_2$ that of
$D$ and $s_3$ that of $t$. By induction hypothesis, we have 
\begin{center}
	$\|\Gamma\| \vdash |C|:U_{s_1}~~and~~~\|\Gamma\|, x:\|C\| \vdash |D|:U_{s_2}$
\end{center}
{\em i.e.}
$$\|\Gamma\|, x:(\varepsilon_{s_1}~|C|) \vdash |D|:U_{s_2}$$
Thus
$$\|\Gamma\| \vdash 
(\dot{\Pi}_{\langle s_1, s_2, s_3 \rangle}~|C|~\lambda
x:(\varepsilon_{s_1}~|C|)~|D|):U_{s_{3}}$$ 
{\em i.e.}
$$\|\Gamma\| \vdash |\Pi x:C~D|:\|s_{3}\|$$
\item If $t = \lambda x:C~u$, then 
we have 
$$\Gamma, x:C \vdash u:D$$
and $B = \Pi x:C~D$. 
By induction hypothesis, we have 
$$\|\Gamma\|, x:\|C\| \vdash |u|:\|D\|$$
{\em i.e.}
$$\|\Gamma\|, x:(\varepsilon_{s_1}~|C|) \vdash |u|:\|D\|~~then~~
\|\Gamma\| \vdash \lambda x:(\varepsilon_{s_1}~|C|)~|u|:\Pi x:\|C\|~\|D\|$$
{\em i.e.}
$$\|\Gamma\| \vdash |t|:\|\Pi x:C~D\|$$
\item If $t = u~v$, then we have 
$$\Gamma \vdash u:\Pi x:C~D,~~~
\Gamma \vdash v:C$$
 and $B = (v/x)D$.
By induction hypothesis, we get 
$$\|\Gamma\| \vdash |u|:\|\Pi x:C~D\| = \Pi x:\|C\|~\|D\|~~~and~~~\|\Gamma\| \vdash |v|:\|C\|$$
Thus
$$\|\Gamma\| \vdash |t|:(|v|/x)\|D\| = \|(v/x)D\|$$
\end{itemize}
\end{proof}

\subsection{Termination}

\begin{proposition}
\label{termination1}
If $\lambda \Pi_{P}$ terminates then $P$ terminates. 
\end{proposition}

\begin{proof}
Let $t_1$ be a well-typed term in $P$ and 
$t_1, t_2, ...$ be a reduction sequence of $t_1$ in $P$. 
By Proposition \ref{correction}, the term $|t_1|$ is well-typed in
$\lambda \Pi_{P}$ and, by Proposition \ref{reduct}, $|t_1|, |t_2|,
...$ is a reduction sequence of $|t_1|$ in $\lambda \Pi_{P}$. Hence it
is finite.
\end{proof}

\subsection{Confluence}

\VC{
\begin{proposition}
For any functional Pure Type System $P$,  the relation $\lra$ is confluent in
$\lambda \Pi_P$
\end{proposition}

Like that of pure $\lambda$-calculus, the reduction relation of
$\lambda \Pi_P$ is not strongly confluent. Thus, we introduce
another reduction relation $(\apa)$ that can reduce, in one step, none
to all the $\beta\cal{R}$-redices that appears in a term, that is
strongly confluent and such that $\apar = ~\lra^{*}$. Then, from the
confluence of the relation $\apa$, we deduce that of the relation
$\lra$. See the long version of the paper for the full proof.}

\VL{We prove in this section that for any functional Pure Type System
$P$, the system $\lambda \Pi_P$ is confluent.  Like that of pure
$\lambda$-calculus, the reduction relation of $\lambda \Pi_P$ is not
strongly confluent: the term $M = (\lambda x~(x~x))~((\lambda y~y)~
0)$ has two one-step reducts: $N_1 = (\lambda x~(x ~x))~0$ and $N_2 =
((\lambda y~y)~ 0)~((\lambda y~y) ~0)$ and these two terms have no
common one-step reduct. Thus, we introduce another reduction relation
$(\apa)$ that can reduce, in one step, none to all the
$\beta\cal{R}$-redices that appears in a term, that is strongly
confluent and such that $\apar = ~\lra^{*}$. Then, from the confluence
of the relation $\apa$, we shall deduce that of the relation $\lra$.

\begin{definition}[Parallel reduction]
\label{parred}
The parallel reduction $(\apa)$ in $\lambda \Pi_P$, is the 
smallest relation on terms that verifies:
\begin{center}

~~~$\irule{} 
        {M \apa M}
        {\mbox{\bf~($\alpha$)}}$
~~~~~~~~~~~~~
$\irule{} 
        {\varepsilon_{s_{2}} ~\dot{s_1} \apa U_{s_{1}}}
        {\mbox{\bf~($\beta$) } ~~  \langle s_1,s_2 \rangle \in {\cal A} }$

$\irule{A \apa A' ~~~M \apa M'} 
        {\lambda x:A~M \apa \lambda x:A'~M'}
        {\mbox{\bf ~($\gamma$)}}$
  ~~~~~~~~~~~~~
$\irule{A \apa A' ~~~ B \apa B'} 
        {\Pi x:A~B \apa \Pi x:A'~B'}
        {\mbox{\bf~($\delta$)}}$

~~~$\irule{M \apa M' ~~~ N \apa N'} 
        {M~N \apa M'~N'}
        {\mbox{\bf ~($\theta_1$)}}$
  ~~~~~~~~~~~~~~~~
$\irule{M \apa M' ~~~ N \apa N'} 
        {(\lambda x:A~M)~N \apa (N'/x)M'}
        {\mbox{\bf~($\theta_2$)}}$
      
\bigskip

$\irule{A \apa A' ~~~~~~ B \apa B'} 
        {\varepsilon_{s_{3}} ~(\dot{\Pi}_{\langle s_1,s_2, s_3 \rangle} ~ A~B) \apa \varepsilon_{s_{3}} ~(\dot{\Pi}_{\langle s_1,s_2, s_3 \rangle} ~ A'~B') }
        {\mbox{\bf ~($\eta_1$)}~~\langle s_1,s_2,s_3 \rangle \in {\cal R}}$

$\irule{A \apa A' ~~~~~~ B \apa B'} 
        {\varepsilon_{s_{3}} ~(\dot{\Pi}_{\langle s_1,s_2, s_3 \rangle} ~ A~B) \apa  \Pi x:(\varepsilon_{s_1}~ A' )~(\varepsilon_{s_2} (B'~x))  }
        {\mbox{\bf ~($\eta_2$)}~~\langle s_1,s_2,s_3 \rangle \in {\cal R}}$

\end{center}

\end{definition}

\begin{proposition}
\label{subst}
For all terms $M$, $M'$, $N$, $N'$ of $\lambda \Pi_P$, 
if $M \apa M'$ and $N \apa N'$, then $(N/x)M \apa (N'/x)M'$.
\end{proposition}

\begin{proof} By induction on M.
\begin{itemize}
\item if $M$ is a variable,
\begin{itemize}

\item[$\star$] if $M = x$, then $M' = M = x$ 
\\(because $M \apa M'$ and the only rule we can apply is $(\alpha)$).
\\ Therefore, $(N/x)M = N \apa N' = (N'/x)M'$.

\item[$\star$] if $M = y \neq x$, then, by the rule $(\alpha)$, $(N/x)M = y \apa y = (N'/x)M'$
\\ (and we conclude by the same way for $M=Type$ and $M=Kind$).
\end{itemize}

\item if there exists terms $A$ and $B$ such that $M = \lambda y:A~B$, then there exists terms $A'$ and $B'$ such that $M' = \lambda y:A'~B'$ with $A \apa A'$ and $B \apa B'$ (because the only rules we can apply to an abstraction are $(\alpha)$ and $(\gamma)$). 
\\ By induction hypothesis, we have $(N/x)A \apa (N'/x)A'$ and $(N/x)B \apa (N'/x)B'$.
\\ Therefore, by $(\gamma)$, 
$$(N/x)M = \lambda y:((N/x)A)~(N/x)B \apa \lambda y:((N'/x)A')~(N'/x)B' = (N'/x)M'$$

\item  if there exists terms $A$ and $B$ such that $M = \Pi y:A~B$, then there exists terms $A'$ and $B'$ such that $M' = \Pi y:A'~B'$ with $A \apa A'$ and $B \apa B'$ (because the only rules we can apply to an abstraction are $(\alpha)$ and $(\delta)$). 
\\ By induction hypothesis, we have $(N/x)A \apa (N'/x)A'$ and $(N/x)B \apa (N'/x)B'$.
\\ Therefore, by $(\delta)$, 
$$(N/x)M = \Pi y:((N/x)A)~(N/x)B \apa \Pi y:((N'/x)A')~(N'/x)B' = (N'/x)M'$$

\item  if there exists terms $P$ and $Q$ such that $M = P~Q$,
\\ if the last rule of the derivation of $M \apa M'$ is:

\begin{itemize}

\item[$(\alpha)$] then $M' = M = P~Q$. 
\\We have $P \apa P$ and $Q \apa Q$, then, by induction hypothesis, $(N/x)P \apa (N'/x)P$ and $(N/x)Q \apa (N'/x)Q$. 
\\Therefore, by $(\theta_1)$, 
\\$(N/x)M = (N/x)P ~(N/x)Q \apa (N'/x)P ~(N'/x)Q = (N'/x)M'$.

\item[$(\beta)$] then there exists a rule $\langle s_1, s_2 \rangle$ such that $M=\varepsilon_{s_2}~\dot{s_1}$ and $M'=U_{s_1}$. 
\\ Therefore, by $(\beta)$, $(N/x)M = \varepsilon_{s_2}~\dot{s_1} \apa U_{s_1} = (N'/x)M'$

\item[$(\theta_1)$] then there exists terms $P'$ and $Q'$ such that $M'=P'~Q'$ with $P \apa P'$ and $Q \apa Q'$. By induction hypothesis, we have $(N/x)P \apa (N'/x)P'$ and $(N/x)Q \apa (N'/x)Q'$. Therefore, by $(\theta_1)$,
\\ $(N/x)M = (N/x)P ~(N/x)Q \apa (N'/x)P' ~(N'/x)Q' = (N'/x)M'$

\item[$(\theta_2)$] then there exists terms $A$, $B$, $B'$, $Q'$ such that $P=\lambda y:A~B$ and $M'=(Q'/y)B'$ with $B \apa B'$ and $Q \apa Q'$. 
\\ By induction hypothesis, we have $(N/x)B \apa (N'/x)B'$ and $(N/x)Q \apa (N'/x)Q'$.
\\ Therefore, by $(\theta_2)$,
\\ $ (N/x)M = (\lambda y:(N/x)A~(N/x)B)~(N/x)Q \apa ((N'/x)Q'/y )~(N'/x)B' \\= (N'/x)((Q'/y)B') = (N'/x)M'$

\item[$(\eta_1)$] then there exists a rule $\langle s_1,s_2,s_3 \rangle$ and terms $A$, $A'$, $B$, $B'$ such that $M = \varepsilon_{s_{3}} ~(\dot{\Pi}_{\langle s_1,s_2, s_3 \rangle} ~ A~B)$ and  $M' = \varepsilon_{s_{3}} ~(\dot{\Pi}_{\langle s_1,s_2, s_3 \rangle} ~ A'~B')$, with $A \apa A'$ and $B \apa B'$.
\\ By induction hypothesis, we have $(N/x)A \apa (N'/x)A'$ and $(N/x)B \apa (N'/x)B'$.
\\ Therefore, by $(\eta_1)$,
\\ $(N/x)M = \varepsilon_{s_{3}} ~(\dot{\Pi}_{\langle s_1,s_2, s_3 \rangle} ~ (N/x)A~~(N/x)B) \apa \varepsilon_{s_{3}} ~(\dot{\Pi}_{\langle s_1,s_2, s_3 \rangle} ~ (N'/x)A'~~(N'/x)B') \\ = (N'/x)M'$.

\item[$(\eta_2)$] then there exists a rule $\langle s_1,s_2,s_3 \rangle$ and terms $A$, $A'$, $B$, $B'$ such that $M = \varepsilon_{s_{3}} ~(\dot{\Pi}_{\langle s_1,s_2, s_3 \rangle} ~ A~B)$ and  $M' = \Pi y:(\varepsilon_{s_1}~A')~(\varepsilon_{s_2}~(B'~y))$, with $A \apa A'$ and $B \apa B'$.
\\ By induction hypothesis, we have $(N/x)A \apa (N'/x)A'$ and $(N/x)B \apa (N'/x)B'$.
\\ Therefore, by $(\eta_2)$,
\\ $(N/x)M = \varepsilon_{s_{3}} ~(\dot{\Pi}_{\langle s_1,s_2, s_3 \rangle} ~ (N/x)A~~(N/x)B) \apa \Pi y:(\varepsilon_{s_1}~(N'/x)A')~(\varepsilon_{s_2}~((N'/x)B'~y)) \\ = (N'/x)M'$.
\end{itemize}
\end{itemize}
\end{proof}

Then, following \cite{Lescanne}, we associate, to each term $t$ of
$\lambda\Pi_{P}$, a term $t^{\dag}$, obtained by reducing in parallel
all its $\beta\cal{R}$-redices.

\begin{definition}
\label{fd}
Let $t$ be a term of $\lambda\Pi_{P}$. We define, by induction on the
structure of $t$, the term $t^{\dag}$ as follows:

\begin{itemize} 
\item $x^{\dag} = x,~~~Type^{\dag} = Type,~~~Kind^{\dag} = Kind$

\item $(\lambda x:A~M)^{\dag} = \lambda x:A^{\dag}~M^{\dag},~~~~~(\Pi x:A~B)^{\dag} = \Pi x:A^{\dag}~B^{\dag}$

\item $((\lambda x : A~M)~N)^{\dag} = (N^{\dag}/x)M^{\dag}$

\item $(\varepsilon_{s_2} ~ \dot{s_1})^{\dag} = ~U_{s_1}$, ~if $\langle s_1,s_2 \rangle \in {\cal A}$,

\item $(\varepsilon_{s_3} (\dot{\Pi}_{\langle s_1,s_2,s_3 \rangle} A ~ B))^{\dag} = \Pi x:(\varepsilon_{s_1} A^{\dag})~(\varepsilon_{s_2} (B^{\dag} x))$, ~if $\langle s_1,s_2,s_3 \rangle \in {\cal R}$,

\item $(M~N)^{\dag} = M^{\dag}N^{\dag}$, otherwise.
\end{itemize}
\end{definition}

\begin{proposition}
\label{dagid}
If $M$ is a term of $\lambda \Pi_P$, then ~$M \apa M^{\dag}$
\end{proposition}

\begin{proof}
By induction on M.

\begin{itemize}

\item $x^{\dag} = x$, ~$Type^{\dag} = Type,$ and $Kind^{\dag} = Kind$, then by the rule $(\alpha)$, we have $x \apa x^{\dag}$, $Type \apa Type^{\dag}$ and $Kind \apa Kind^{\dag}$

\item If we suppose, by induction hypothesis, $A \apa A^{\dag}$ and $N
  \apa N^{\dag}$, then, by the rule $(\gamma)$, $\lambda x:A~N\apa
  \lambda x : A^{\dag}~N^{\dag} = (\lambda x:A~N)^{\dag}$

\item If $A \apa A^{\dag}$ and $B \apa B^{\dag}$, then by the rule $(\delta)$, 
\\ $\Pi x:A~B \apa \Pi x:A^{\dag}~B^{\dag} = (\Pi x:A~B)^{\dag}$

\item If $M$ is an application then we consider four cases.

\begin{itemize}

\item[$\star$] If there exists an axiom $\langle s_1,s_2 \rangle$ such
that $M = \varepsilon_{s_2}~\dot{s_1}$, then, by the rule $(\beta)$,
we have $\varepsilon_{s_2}~\dot{s_1} \apa U_{s_1} =
(\varepsilon_{s_2}~\dot{s_1})^{\dag}$.

\item[$\star$] If there exists terms $A$, $N$ and $P$ such that $M =
(\lambda x:A~N)~P$, and if we suppose, by induction hypothesis that $N
\apa N^{\dag}$ and $P \apa P^{\dag}$, then by the rule $(\theta_2)$,
we have $(\lambda x:A~N)~P \apa (P^{\dag}/x)N^{\dag} =((\lambda
x:A~N)~P)^{\dag}$.

\item[$\star$] If there exists a rule $\langle s_1,s_2,s_3\rangle$ and
terms $A$ and $B$ such that \\$M = \varepsilon_{s_{3}}
~(\dot{\Pi}_{\langle s_1,s_2, s_3 \rangle} ~ A~B)$, and if we suppose,
by induction hypothesis that $A \apa A^{\dag}$ and $B \apa B^{\dag}$,
then by the rule $(\eta_2)$, we have\\ $\varepsilon_{s_{3}}
~(\dot{\Pi}_{\langle s_1,s_2, s_3 \rangle} ~ A~B) \apa \Pi
x:(\varepsilon_{s_1}~ A^{\dag})~(\varepsilon_{s_2} (B^{\dag}~x)) =
(\varepsilon_{s_{3}} ~(\dot{\Pi}_{\langle s_1,s_2, s_3 \rangle} ~ A~B)
)^{\dag}$.

\item[$\star$] Otherwise, if there exists terms $N$ and $P$ such that
$M = N~P$, and if we suppose, by induction hypothesis, that $N \apa
N^{\dag}$ and $P \apa P^{\dag}$, then by the rule $(\theta_1)$ we have
$N~P \apa N^{\dag}~P^{\dag} = (N~P)^{\dag}$.
\end{itemize}
\end{itemize}
\end{proof}

\begin{proposition}
\label{ferme}
For all terms $M$ and $M'$ of $\lambda \Pi_P$, 
if ~$M \apa M'$~ then ~$M' \apa M^{\dag}$
\end{proposition}

\begin{proof}
By induction on the last rule of the derivation of $M \apa M'$.
\\ If the last rule is:

\begin{itemize}

\item[$(\alpha)$] then $M'=M$. By Proposition \ref{dagid} we have $M \apa M^{\dag}$

\item[$(\beta)$] then there exists a rule $\langle s_1,s_2 \rangle$ such that $M = \varepsilon_{s_2}~\dot{s_1}$ and $M'=U_{s_1}$. Therefore $M' \apa U_{s_1} = M^{\dag}$ by the rule $(\alpha)$.

\item[$(\gamma)$] then there exists terms $A$, $A'$, $P$ and $P'$ such that $M = \lambda x:A~P$ and $M' = \lambda x:A'~P'$ with $A \apa A'$ and $P \apa P'$. By induction hypothesis, we have $A' \apa A^{\dag}$ and $P' \apa P^{\dag}$, then, by the rule $(\gamma)$, 
\\$ M' =  \lambda x:A'~P' \apa \lambda x:A^{\dag}~P^{\dag} = M^{\dag}$

\item[$(\delta)$] then there exists terms $A$, $A'$, $B$ and $B'$ such that $M = \Pi x:A~B$ and $M' = \Pi x:A'~B'$ with $A \apa A'$ and $B \apa B'$. By induction hypothesis, we have $A' \apa A^{\dag}$ and $B' \apa B^{\dag}$, then, by the rule $(\delta)$, 
\\$ M' =  \Pi x:A'~B' \apa \Pi x:A^{\dag}~B^{\dag} = M^{\dag}$

\item[$(\theta_1)$]  then there exists terms $P$, $P'$, $Q$ and $Q'$ such that $M=P~Q$ and $M'=P'~Q'$ with $P \apa P'$ and $Q \apa Q'$.

\begin{itemize}

\item[$\star$] If there exists terms $A$ and $B$ such that $P = \lambda x:A~B$, then there exists terms $A'$ and $B'$ such that $P' = \lambda x:A'~B'$ with $A \apa A'$ and $B \apa B'$ (because the only rules we can apply to an abstraction are $(\alpha)$ and $(\gamma)$). Therefore, by induction hypothesis, $B' \apa B^{\dag}$ and $Q' \apa Q^{\dag}$. 
\\And, by $(\theta_2)$, 
\\$M'=P'~Q' = ( \lambda x:A'~B')~Q' \apa (Q^{\dag}/x)B^{\dag} = ((\lambda x:A~B)~Q)^{\dag} = M^{\dag}$

\item[$\star$] If there exists a axiom $\langle s_1,s_2\rangle$ such that $P = \varepsilon_{s_2}$ and $Q = \dot{s_1}$, then $P'=P$ and $Q'=Q$ (because the only rules we can apply to $(\varepsilon_{s_2}~ \dot{s_1})$ are $(\alpha)$, $(\beta)$ and $(\theta_1$ with $(\alpha)$ on both premises)).
\\And, by $(\beta)$,  $M'=(\varepsilon_{s_2}~ \dot{s_1}) \apa U_{s_1} = M^{\dag}$

\item[$\star$] If there exists a rule $\langle s_1,s_2,s_3\rangle$ and terms $A$, $B$, such that $P = \varepsilon_{s_3}$ and $Q = \dot{\Pi}_{\langle s_1,s_2,s_3\rangle} A~B$, then $P'=P=\varepsilon_{s_3}$ (because the only rule we can apply is $(\alpha)$), and there exists terms $A'$ and $B'$ such that 
\\$Q' = \dot{\Pi}_{\langle s_1,s_2,s_3\rangle} A'~B'$ with $A \apa A'$ and $B \apa B'$ (because the only rule we can apply is $(\theta_1)$).
\\ Therefore, by induction hypothesis, $A' \apa A^{\dag}$ and $B' \apa B^{\dag}$. 
\\ And, by $(\eta_2)$, 
\\$M' = \varepsilon_{s_3}~ (\dot{\Pi}_{\langle s_1,s_2,s_3\rangle} A'~B') \apa \Pi x:(\varepsilon_{s_1} ~A^{\dag}) (\varepsilon_{s_2} ~(B^{\dag}~x)= M^{\dag}$

\item[$\star$] Otherwise, $(P~Q)^{\dag} = P^{\dag}~Q^{\dag}$. We have, by induction hypothesis, $P' \apa P^{\dag}$ and $Q' \apa Q^{\dag}$. Therefore, by $(\theta_1)$, $M' = P'~Q' \apa P^{\dag}~Q^{\dag} = M^{\dag}$.

\end{itemize}

\item[$(\theta_2)$] then there exists terms $A$, $B$, $B'$, $Q$, $Q'$ such that $M = (\lambda x:A ~B)~Q$ and $M' = (Q'/x)B'$ with $B \apa B'$ and $Q \apa Q'$.
\\ By induction hypothesis, we have $B' \apa B^{\dag}$ and $Q' \apa Q^{\dag}$.
\\ Therefore, by Proposition \ref{subst},
$M'= (Q'/x)B' \apa (Q^{\dag}/x)B^{\dag} = M^{\dag}$

\item[$(\eta_1)$] then there exists a rule $\langle s_1,s_2,s_3 \rangle$ and terms $A$, $B$ such that 
\\$M = \varepsilon_{s_{3}} ~(\dot{\Pi}_{\langle s_1,s_2, s_3 \rangle} ~ A~B)$ and $M' = \varepsilon_{s_{3}} ~(\dot{\Pi}_{\langle s_1,s_2, s_3 \rangle} ~ A'~B')$ 
\\with $A \apa A'$ and $B \apa B'$.
\\ By induction hypothesis, we have $A' \apa A^{\dag}$ and $B' \apa B^{\dag}$.
\\ Therefore, by $(\eta_2)$,
\\ $M' = \varepsilon_{s_{3}} ~(\dot{\Pi}_{\langle s_1,s_2, s_3 \rangle} ~ A'~B') \apa \Pi x:(\varepsilon_{s_1} A^{\dag})~(\varepsilon_{s_2} (B^{\dag} x)) = M^{\dag}$.

\item[$(\eta_2)$] then there exists a rule $\langle s_1,s_2,s_3 \rangle$ and terms $A$, $B$ such that 
\\$M = \varepsilon_{s_{3}} ~(\dot{\Pi}_{\langle s_1,s_2, s_3 \rangle} ~ A~B)$ and $M' = \Pi x:(\varepsilon_{s_1} A')~(\varepsilon_{s_2} (B' x))$ 
\\with $A \apa A'$ and $B \apa B'$.
\\ By induction hypothesis, we have $A' \apa A^{\dag}$ and $B' \apa B^{\dag}$.
\\ Therefore, by $(\alpha)$, $(\delta)$ and $(\eta_1)$,
\\ $M' = \Pi x:(\varepsilon_{s_1} A')~(\varepsilon_{s_2} (B' x)) \apa \Pi x:(\varepsilon_{s_1} A^{\dag})~(\varepsilon_{s_2} (B^{\dag} x)) = M^{\dag}$.
\end{itemize}
\end{proof}

\begin{proposition}
\label{conf}
The relation $\apa$ is strongly confluent in $\lambda \Pi_P$, 
{\em i.e.}
for all $M$,
$M'$ and $M''$, if $M \apa M'$ and $M \apa M''$ then there
exists a term $N$ such that $M' \apa N$ and $M'' \apa N$.
\end{proposition}

\begin{proof}
By Proposition \ref{ferme}, $M' \apa M^{\dag}$ and $M'' \apa M^{\dag}$.
\end{proof}

\begin{proposition}
\label{parstar}
For all terms $M$ and $M'$ of $\lambda \Pi_P$,
\begin{enumerate}
\item if $M \lra M'$ then $M \apa M'$
\item if $M \apa M'$ then $M \lra^{*} M'$
 \item ~$M \apar M' $~if and only if~$M \lra^{*} M'$  ~~~({\em i.e.}
 $\apar = ~~ \lra^{*} $). 
\end{enumerate}
\end{proposition}

\begin{proof}~

\begin{enumerate}

\item If $M \lra M'$, then $M \lra_{\beta} M'$ or $M \lra_{\cal{R}} M'$

\begin{itemize}
\item[$\star$] If $M \lra_{\beta} M'$ then $M \apa M'$, by $(\theta_2)$ and $(\alpha)$
\item[$\star$] If $M \lra_{\cal{R}} M'$ then $M \apa M'$, by $(\beta)$, or $(\eta_2)$ and $(\alpha)$
\end{itemize}

\item By induction on the last rule of the derivation of $M \apa M'$.
\\ If the last rule is:

\begin{itemize}

\item[$(\alpha)$] then $M'=M$, and we have $M \lra^{*} M$.

\item[$(\beta)$] then there exists a rule $\langle s_1,s_2 \rangle$ such that $M = \varepsilon_{s_2}~\dot{s_1}$ and $M'=U_{s_1}$, and we have $ \varepsilon_{s_2}~\dot{s_1} \lra_{\cal{R}} U_{s_1}$, therefore $M \lra^{*} M'$.

\item[$(\gamma)$] then there exists terms $A$, $A'$, $P$ and $P'$ such that $M = \lambda x:A~P$, $M' = \lambda x:A'~P'$ with $A \apa A'$ and $P \apa P'$. 
\\By induction hypothesis, we have $A \lra^{*} A'$ and $P \lra^{*} P'$, 
\\therefore $M = \lambda x:A~P \lra^{*} \lambda x:A'~P'=M'$.

\item[$(\delta)$] then there exists terms $A$, $A'$, $B$ and $B'$ such that $M = \Pi x:A~B$, $M' = \Pi x:A'~B'$ with $A \apa A'$ and $B \apa B'$. 
\\By induction hypothesis, we have $A \lra^{*} A'$ and $B \lra^{*} B'$, 
\\therefore  $ M = \Pi x:A~B \lra^{*} \Pi x:A'~B' = M'$.

\item[$(\theta_1)$]  then there exists terms $P$, $P'$, $Q$ and $Q'$ such that $M=P~Q$ and $M'=P'~Q'$ with $P \apa P'$ and $Q \apa Q'$.
\\By induction hypothesis, we have $P \lra^{*} P'$ and $Q \lra^{*} Q'$, 
\\therefore  $ M = P~Q \lra^{*} P'~Q' = M'$

\item[$(\theta_2)$] then there exists terms $A$, $B$, $B'$, $Q$, $Q'$ such that $M = (\lambda x:A ~B)~Q$ and $M' = (Q'/x)B'$ with $B \apa B'$ and $Q \apa Q'$.
\\ By induction hypothesis, we have $B \lra^{*} B'$ and $Q \lra^{*} Q'$,
\\ therefore
$M = (\lambda x:A ~B)~Q \lra^{*} (\lambda x:A ~B')~Q' \lra_{\beta} (Q'/x)B' = M'$.

\item[$(\eta_1)$] then there exists a rule $\langle s_1,s_2,s_3 \rangle$ and terms $A$, $B$ such that 
\\$M = \varepsilon_{s_{3}} ~(\dot{\Pi}_{\langle s_1,s_2, s_3 \rangle} ~ A~B)$ and $M' = \varepsilon_{s_{3}} ~(\dot{\Pi}_{\langle s_1,s_2, s_3 \rangle} ~ A'~B')$ 
\\with $A \apa A'$ and $B \apa B'$.
\\ By induction hypothesis, we have $A \lra^{*}  A'$ and $B \lra^{*}  B'$,
\\ therefore $M = \varepsilon_{s_{3}} ~(\dot{\Pi}_{\langle s_1,s_2, s_3 \rangle} ~ A~B) \lra^{*} \varepsilon_{s_{3}} ~(\dot{\Pi}_{\langle s_1,s_2, s_3 \rangle} ~ A'~B') = M'$.

\item[$(\eta_2)$] then there exists a rule $\langle s_1,s_2,s_3 \rangle$ and terms $A$, $B$ such that 
\\$M = \varepsilon_{s_{3}} ~(\dot{\Pi}_{\langle s_1,s_2, s_3 \rangle} ~ A~B)$ and $M' = \Pi x:(\varepsilon_{s_1} A')~(\varepsilon_{s_2} (B' x))$ 
\\with $A \apa A'$ and $B \apa B'$.
\\ By induction hypothesis, we have $A \lra^{*} A'$ and $B \lra^{*} B'$, 
 therefore 
 \\$M = \varepsilon_{s_{3}} ~(\dot{\Pi}_{\langle s_1,s_2, s_3 \rangle} ~ A~B) \lra^{*} \varepsilon_{s_{3}} ~(\dot{\Pi}_{\langle s_1,s_2, s_3 \rangle} ~ A'~B') 
 \\ \lra_{\cal{R}} \Pi x:(\varepsilon_{s_1} A')~(\varepsilon_{s_2} (B' x)) = M'$.
\end{itemize}

\item By induction on the number of reductions in $M \lra^{*} M'$ and the first point, for one way. And by induction on the length of the derivation of $M \apa M'$ and the second point for the other way.

\end{enumerate}
\end{proof}

\medskip

\begin{proposition}
The relation $\lra$ is confluent in $\lambda \Pi_P$, {\em i.e.} for all $M$,
$M'$ and $M''$, if $M \lra^{*} M'$ and $M \lra^{*} M''$ then there
exists a term $N$ such that $M' \lra^{*} N$ and $M'' \lra^{*} N$.
\end{proposition}

\begin{proof} From Proposition \ref{conf} the relation 
$\apa$ is strongly confluent, hence it is confluent. Hence, by 
Proposition \ref{parstar} the relation $\lra$ is confluent.
\end{proof}
}

\section{Conservativity}

Let $P$ be a functional Pure Type System.  We could attempt to prove
that if the type $\|A\|$ is inhabited in $\lambda \Pi_P$, then $A$ is
inhabited in $P$, and more precisely that if $\Gamma$ is a context and
$A$ a term in $P$ and $t$ a term in $\lambda \Pi_P$, such that
$\|\Gamma\| \vdash t : \|A\| $, then there exists a term $u$ of $P$
such that $|u| = t$ and $\Gamma \vdash u : A$. Unfortunately this
property does not hold in general as shown by the following
counterexamples.

\begin{example} 
If $P$ is the simply-typed lambda-calculus, then 
the polymorphic identity is not well-typed in $P$, in particular:
\\$~~~~~~nat:Type~\nvdash ~((\lambda X:Type~\lambda x:X~x)~nat) : (nat
  \Ra nat)$
 \\however, in $\lambda \Pi$, we have 
\\$~~~~~~nat : \|Type\| \vdash ((\lambda X:\|Type\|~\lambda
   x:\|X\|~x)~|nat|) :\|nat \Ra nat\|$.
  \end{example}

\begin{example} 
If $\langle s_1,s_2,s_3 \rangle \in R$, 
$\Sigma_P \vdash \dot{\Pi}_{\langle s_1,s_2,s_3 \rangle} : \| \Pi X:s_1~((X \Ra s_2) \Ra s_3)\|$
but the term $\dot{\Pi}_{\langle s_1,s_2,s_3 \rangle}$ is not the
translation of any term of $P$. 
\end{example}

Therefore, we shall prove a slightly weaker property: that 
if the type $\|A\|$ is inhabited {\em by a normal term} in $\lambda \Pi_P$, 
then $A$ is inhabited in $P$. Notice that this
restriction vanishes if $\lambda \Pi_P$ is terminating.

We shall prove, in a first step, that if $\|\Gamma\| \vdash t :
\|A\|$, and $t$ is a weak $\eta$-long normal term then there exists a
term in $u$ such that such that $|u| = t$ and $\Gamma \vdash u : A$. Then we
shall get rid of this restriction on weak $\eta$-long forms.

\begin{definition}
\label{weak}
A term $t$ of $\lambda \Pi_P$ is a weak $\eta$-long term if and only
if each occurrence of $ \dot{\Pi}_{\langle s_1,s_2,s_3 \rangle}$ in t,
is in a subterm of the form $(\dot{\Pi}_{\langle s_1,s_2,s_3
\rangle}~t_1~t_2)$ ({\em i.e.} each occurrence of $ \dot{\Pi}_{\langle
s_1,s_2,s_3 \rangle}$ is $\eta$-expanded).
\end{definition}

\begin{definition}[Back translation] We suppose that $P$ contains at
least one sort: $s_0$.  
Then we define a translation from $\lambda \Pi_{P}$ to $P$ as follows:
\begin{itemize}
\item $x^* = x$, 
~~~~~ $s^* = s_0$  
~~~~~ $\dot{s}^* = s$, 
~~~~~ $U_s^* = s$, 
\item $(\Pi x:A~B)^* = \Pi x:A^*~B^*$, 
\item $(\lambda x:A~t)^* = \lambda x:A^*~t^*$, 
\item $(\dot{\Pi}_{\langle s_1, s_2, s_3 \rangle}~A~B)^* = \Pi x:A^*~(B^*~x)$, 
\item $(\varepsilon_s~u)^* = u^*$, 
\item $(t~u)^* = t^*~u^*$ otherwise.
\end{itemize}
\end{definition}

\begin{proposition}
\label{idp}
The back translation $(.)^{*}$ is a right inverse of $|.|$ and $\|.\|$
{\em i.e.} for all $t$ such that $|t|$ ({\em resp.} $\|t\|$) is well defined, $|t|^{*} = t$ ({\em resp.} $\|t\|^{*} = t$).
\end{proposition}

\begin{proof}
By induction on the structure of $t$.
\end{proof}

\begin{proposition}
\label{backreduct} For all $t$, $u$ terms and $x$ variable of $\lambda \Pi_P$,
\begin{enumerate} 
\item $((u/x)t)^* = (u^*/x)t^*$
\item If $t \lra u$ then $t^* \lra_{\beta}^* u^*$ in $P$.
\end{enumerate}
\end{proposition}

\begin{proof}
\begin{enumerate}
\item By induction on $t$.
\item If $t \lra_{\beta} u$ then $t^* \lra_{\beta} u^*$, and if $t \lra_{{\cal R}} u$, then $t^* = u^*$.
\end{enumerate}
\end{proof}

\medskip 

\begin{proposition}
\label{equiv} For all terms $A$, $B$ of $P$ and $C$, $D$ of $\lambda \Pi_P$ 
(such that $\|A\|$ and $\|B\|$ are well defined), 
\begin{enumerate}
\item If $A \equiv_{_{\beta}} B $, then $\|A\| ~\equiv \|B\|$.
\item If $C \equiv D $, then $C^* ~\equiv_{_{\beta}} D^*$.
\item If $\|A\| \equiv \|B\| $, then $A ~\equiv_{_{\beta}} B$.
\end{enumerate}
\end{proposition}

\begin{proof} 
\begin{enumerate}
\item By induction on the length of the path of $\beta$-reductions and
  $\beta$-expansions between $A$ and $B$, and by Proposition \ref{reduct}.
\item By the same reasoning as for the first point, using Proposition \ref{backreduct}.
\item By the first and second points and Proposition \ref{idp}.

\end{enumerate}
\end{proof}

\begin{proposition}[Conservativity]
\label{conservativity} If there exists a context $\Gamma$, a term $A$
of $P$, and a term $t$, in weak $\eta$-long normal form, of $\lambda \Pi_P$, such that
 $\|\Gamma\| \vdash t : \|A\| $,
 Then there exists a term $u$ of $P$ such that~
  $|u|~ \equiv t $~ and ~$  \Gamma \vdash u : A $.
\end{proposition}

\begin{proof} By induction on $t$.
\begin{itemize}
\item[$\bullet$] If $t$ is a well-typed product or sort, then it cannot
be typed by a translated type (by confluence of $\lambda \Pi_P$). 

\item[$\bullet$] If $t=\lambda x:B~t'$. 
The term $t$ is well typed, thus there exists a term $C$ of $\lambda \Pi_P$, such that $~\|\Gamma\| \vdash t : \Pi x:B~C~$ ($\alpha_0$). 

Therefore 
 $\|A\| \equiv \Pi x:B~C$   ($\alpha_1$),
and $ A \equiv \|A\|^{*} \equiv (\Pi x:B~C)^{*} = \Pi x:B^{*}~C^{*} $ ($\alpha_2$). 

$\|A\|$ is well defined, and $A$ cannot be a sort by  ($\alpha_2$) and confluence of $\lambda \Pi_P$, then $A$ is well-typed.

Therefore, by ($\alpha_2$), confluence of $\lambda \Pi_P$ and subject-reduction of $\lra_{\beta}$, there exists 
terms $B^{\sharp}$ and $C^{\sharp}$ and sorts $s_{B^{\sharp}}$, $s_{C^{\sharp}}$, $s_3$ of $P$, such that 
both $A$ and $\Pi x:B^{*}~C^{*}$ reduce to $ \Pi x:B^{\sharp}~C^{\sharp} $ ($\beta$), with  $\Gamma \vdash B^{\sharp} : s_{B^{\sharp}}$, and $\Gamma, B^{\sharp} : s_{B^{\sharp}} \vdash C^{\sharp} : s_{C^{\sharp}}$ where $\langle s_{B^{\sharp}}, s_{C^{\sharp}}, s_3 \rangle$ is a rule of $P$. 

In particular, $\| B^{\sharp} \|$, $\| C^{\sharp} \|$ and $\|\Pi x:B^{\sharp}~C^{\sharp}\| $ are well defined. 

Moreover, $ A \equiv  \Pi x:B^{\sharp}~C^{\sharp} $ by ($\beta$), then, by Propositions \ref{equiv} and \ref{Piequiv}, $\|A\| \equiv \|\Pi x:B^{\sharp}~C^{\sharp}\| \equiv \Pi x:\|B^{\sharp}\|~\|C^{\sharp}\|$. 

Therefore $\Pi x:B~C \equiv \Pi x:\|B^{\sharp}\|~\|C^{\sharp}\|$ by ($\alpha_1$), and by confluence of $\lambda \Pi_P$, we have $B \equiv \|B^{\sharp}\|$ ($\gamma_0$).

Then, by ($\alpha_0$), $~\|\Gamma\| \vdash t : \Pi x:\|B^{\sharp}\|~\|C^{\sharp}\|$ and $\|\Gamma\|, x:\|B^{\sharp}\| \vdash t' : \|C^{\sharp}\| $.

The term $\lambda x : B~t'$ is in weak $\eta$-long normal form, thus
$t'$ is also in weak $\eta$-long normal form, and, by induction
hypothesis, there exists a term $u'$ of $P$, such that $ |u'| \equiv t' $ ($\gamma_1$)  and $\Gamma, x:B^{\sharp} \vdash u' : C^{\sharp}$ ($\gamma_2$).

Let $u = \lambda x:B^{\sharp}~u'$.

By ($\gamma_2$), we have ~$ \Gamma \vdash u~:~\Pi x: B^{\sharp}~ C^{\sharp}$, and $\Gamma \vdash u~:~A$ by ($\beta$).

 And ~$ |u| = |\lambda x:B^{\sharp}~u'| = \lambda x:\|B^{\sharp}\|~|u'| \equiv \lambda x:B~t' = t$ by ($\gamma_0$) and ($\gamma_2$).

\medskip

\item[$\bullet$] If $t$ is an application or a variable,
as it is normal, it has the form $~x~t_1~...~t_n~$ for some variable
$x$ and terms $t_1$, ..., $t_n$.  We have  $\|\Gamma\| \vdash x~t_1~...~t_n : \|A\| ~~~(\alpha_0)$.

\begin{itemize}

\item[$\multimap$] If $x$ is a variable of the context $\Sigma_P$,

\begin{itemize}

\item If $~x=\dot{s_1}~$ (where $\langle s_1,s_2 \rangle$ is an axiom of $P$), 
\\then $n=0$ (because $t$ is well typed) and $\|A\| = U_{s_2}$. 
\\ We have $~\vdash s_1 : s_2$ in $P$, therefore $~\Gamma \vdash s_1 : s_2$.

\item If $~x=U_s~$ (where $s$ is a sort of $P$), then $n=0$ and $\|A\|
\equiv Type$. That's an absurdity by confluence of $\lambda
\Pi_P$.

\item If $~x=\varepsilon_s~$ (where $s$ is a sort of $P$), then, as
$t$ is well typed $n \leq 1$.

\begin{itemize}

\item[$\star$] If $n=1$, then $\|\Gamma\| \vdash t_1 : U_s$, and
$\|A\| \equiv Type$ (absurdity). 

\item[$\star$] If $n=0$, then $\|A\| \equiv U_s \Ra Type$, thus by
Propositions \ref{equiv} and \ref{Piequiv}, \\$~U_s \Ra Type \equiv \|( U_s
\Ra Type)^{*}\| = \|s \Ra s_0\| \equiv \|s\| \Ra \|s_0\|$.
\\Therefore $~Type \equiv \|s_0\|~$ (absurdity).
\end{itemize}

\item If $~x=\dot{\Pi}_{\langle s_1, s_2, s_3 \rangle}~$ (where
$\langle s_1, s_2, s_3 \rangle$ is a rule of $P$), then as $t$ is
well-typed and in weak $\eta$-long form, $n = 2$. 
We have $~ \|A\| \equiv U_{s_3}~$ thus $~ A \equiv s_3~$ 
by Proposition \ref{equiv}.
 \\~~ And $~\|\Gamma\| \vdash t_1 : U_{s_1} ~~~ {\em i.e.} ~~~ \|\Gamma\| \vdash t_1 : \|s_1\| $.
\\ ~~And $~\|\Gamma\|, t_1 : U_{s_1} \vdash t_2 : ((\varepsilon_{s_1} t_1) \Ra U_{s_2}) ~~~  (\alpha_1)$

$t_1$ is also in weak $\eta$-long normal form, then, by induction hypothesis, there exists a term $u_1$ of $P$ such that:
\begin{center}
$ |u_1| \equiv t_1  ~~~ and ~~~ \Gamma \vdash u_1 : s_1 ~~(\beta_1)$
\end{center}

Then, by $(\alpha_1)$, 
$~ \|\Gamma\|, t_1 : \|s_1\| \vdash t_2 : \| u_1 \Ra s_2\| $.
\\In particular, $~~ \|\Gamma\|, t_1 : \|s_1\| \vdash t_2 : \| u_1 \| \Ra \|s_2\| $.

However $t_2$ is also in weak $\eta$-long normal form, then there exists a term $t'_2$ (in weak $\eta$-long normal form) of $\lambda \Pi_P$ such that
 \begin{center}
 $t_2 = \lambda x: \| u_1 \| ~t'_2 ~~~ and ~~~  \|\Gamma\|, x : \|u_1\| \vdash t'_2 : \|s_2\| $
 \end{center}

By induction hypothesis, there exists a term $u'_2$ of $P$, such that
\begin{center}
$   |u'_2| \equiv t'_2  ~~~ and ~~~ \Gamma, x : u_1 \vdash u'_2 : s_2 ~~(\beta_2)$
\end{center}

Then we choose $~ u = \Pi x:u_1~u'_2~$ that verifies $~ \Gamma \vdash u : s_3 ~$,
by $(\beta_1)$, $(\beta_2)$, and the fact that $\langle s_1,s_2,s_3 \rangle$ is a rule of $P$. And, finally,

$|u| = \dot{\Pi}_{\langle s_1, s_2, s_3 \rangle} |u_1| ~(\lambda x: (\varepsilon_{s_1} |u_1|)~|u'_2| \equiv \dot{\Pi}_{\langle s_1, s_2, s_3 \rangle} t_1 ~t_2 = t$

\end{itemize}

\medskip

\item[$\multimap$] If $x$ is a variable of the context $\Gamma$,

For $k \in \{0,..,n\}$, let $(H_k)$ be the statement: 
\begin{quote}
There exists a term $T_k$ of $P$, such that $\|\Gamma\| \vdash x~t_1~...~t_k ~:~ \|T_k\|$.
\end{quote}

We first prove $(H_{0})$,...,$(H_{n})$ by induction.

\begin{itemize}

\item[$\star$] $k=0$ : $x$ is a variable of the context $\Gamma$, then there exists a well typed term or a sort $T$ in $P$ such that $\Gamma$ contains $x:T$. Therefore $\|\Gamma\|$ contains $x:\|T\|$.

\item[$\star$] $0 \leq k \leq n-1$ : We suppose $(H_k)$.

$x~t_1~...~t_{k+1}~$ is well typed in $\Gamma$, then there exists terms $D$ and $E$ of $\lambda \Pi_P$ such that :

$~\|\Gamma \| \vdash t_{k+1} : D  ~~(\delta_1) ~$, 
$~     \|\Gamma \| \vdash x~t_1~...~t_k : \Pi y:D~E  ~~(\delta_2)$, and $~     \|\Gamma \| \vdash x~t_1~...~t_{k+1} : (t_{k+1}/y)E  ~~(\delta_3)$.

However, by $(H_k)$, there exists $T_k$, such that $\|\Gamma\| \vdash x~t_1~...~t_k ~:~ \|T_k\|$, therefore $\|T_k\| \equiv \Pi y:D~E$ ($\delta_4$), by ($\delta_2$), and $T_k \equiv \Pi y:D^{*}~E^{*} $ ($\delta_5$),  by Propositions \ref{idp} and \ref{equiv}.

$\|T_k\|$ is well defined and $T_k$ can't be a sort by ($\delta_5$) then $T_k$ is well typed.
Then, by ($\delta_5$), confluence of $\lambda \Pi_P$ and subject-reduction of $\lra_{\beta}$, there exists 
terms $D^{\sharp}$ and $E^{\sharp}$ and sorts $s_{D^{\sharp}}$, $s_{E^{\sharp}}$, $s_3$ of $P$, such that 
both $T_k$ and $\Pi x:D^{*}~E^{*}$ reduce to $ \Pi x:D^{\sharp}~E^{\sharp} $ ($\delta_6$), with  $\Gamma \vdash D^{\sharp} : s_{D^{\sharp}}$, and $\Gamma, D^{\sharp} : s_{D^{\sharp}} \vdash E^{\sharp} : s_{E^{\sharp}}$ where $\langle s_{D^{\sharp}}, s_{E^{\sharp}}, s_3 \rangle$ is a rule of $P$. 

In particular, $\| D^{\sharp} \|$, $\| E^{\sharp} \|$ and $\|\Pi x:D^{\sharp}~E^{\sharp}\| $ are well defined. 

Moreover, $ T_k \equiv  \Pi x:D^{\sharp}~E^{\sharp} $ by ($\delta_6$), then, by Propositions \ref{equiv} and \ref{Piequiv}, $\|T_k\| \equiv \|\Pi x:D^{\sharp}~E^{\sharp}\| \equiv \Pi x:\|D^{\sharp}\|~\|E^{\sharp}\|$. 

In particular, $~ E \equiv \|E^{\sharp}\|$, by ($\delta_4$) and confluence of $\lambda \Pi_P$. 

Moreover, $t_{k+1}$ is in weak $\eta$-long form, then by induction hypothesis, there exists a term $u_{k+1}$ of $P$ such that $|u_{k+1}| \equiv t_{k+1}$.

Finally, $(t_{k+1}/y)E \equiv (|u_{k+1}|/y)E \equiv (|u_{k+1}|/y)\|E^{\sharp}\| \equiv \| (u_{k+1}/y)E^{\sharp} \|$.

 And we conclude, by $(\delta_3)$ and the conversion rule of $\lambda \Pi_P$.

\end{itemize}

Then, if $n=0$, we take $u = x$ and $\Gamma$ contains $x:T$ with
$\|T\| \equiv \|A\|$.\\ 
And, if $n>0$, then, by $(\alpha_0)$, there exists terms $B$ and $C$ of
$\lambda \Pi_P$ such that  $~\|\Gamma \| \vdash t_{n} : B
~~(\theta_1)$ and  $~ \|\Gamma \| \vdash x~t_1~...~t_{n-1} : \Pi y:B~C
~~  (\theta_2)~$ 
with $~\|A\| \equiv (t_{n}/y)C  ~~  (\theta_3)~$.

By ($H_{n-1}$), there exists $T_{n-1}$, such that $~ \|\Gamma \| \vdash x~t_1~...~t_{n-1} : \|T_{n-1}\|$, therefore $\|T_{n-1}\| \equiv \Pi y:B~C$ by ($\theta_2$), and $T_{n-1} \equiv \Pi y:B^{*}~C^{*}$ ($\theta_4$).

$\|T_{n-1}\|$ is well defined, and $T_{n-1}$ cannot be a sort by  ($\theta_4$), then $T_{n-1}$ is well-typed.

Therefore, by ($\theta_2$), confluence of $\lambda \Pi_P$ and subject-reduction of $\lra_{\beta}$, there exists 
terms $B^{\sharp}$ and $C^{\sharp}$ and sorts $s_{B^{\sharp}}$, $s_{C^{\sharp}}$, $s_3$ of $P$, such that 
both $T_{n-1}$ and $\Pi x:B^{*}~C^{*}$ reduce to $ \Pi x:B^{\sharp}~C^{\sharp} $ ($\mu_0$), with  $\Gamma \vdash B^{\sharp} : s_{B^{\sharp}}$, and $\Gamma, B^{\sharp} : s_{B^{\sharp}} \vdash C^{\sharp} : s_{C^{\sharp}}$ where $\langle s_{B^{\sharp}}, s_{C^{\sharp}}, s_3 \rangle$ is a rule of $P$. 

In particular, $\| B^{\sharp} \|$, $\| C^{\sharp} \|$ and $\|\Pi x:B^{\sharp}~C^{\sharp}\| $ are well defined. 

Moreover, $ T_{n-1} \equiv  \Pi x:B^{\sharp}~C^{\sharp} $ by ($\mu_0$), then, by Propositions \ref{equiv} and \ref{Piequiv}, $\|T_{n-1}\| \equiv \|\Pi x:B^{\sharp}~C^{\sharp}\| \equiv \Pi x:\|B^{\sharp}\|~\|C^{\sharp}\|$. 

Therefore $\Pi x:B~C \equiv \Pi x:\|B^{\sharp}\|~\|C^{\sharp}\|$, and by confluence of $\lambda \Pi_P$, we have $B \equiv \|B^{\sharp}\|$ ($\gamma_0$).

Thus, $~\|\Gamma \| \vdash t_{n} : \|B^{\sharp}\|    ~$ and $~      \|\Gamma \| \vdash x~t_1~...~t_{n-1} : \| \Pi y : B^{\sharp} ~C^{\sharp}\|  $.

$t_{n}$ and $x~t_1~...~t_{n-1}$ are both in weak $\eta$-long normal form, then, by induction hypothesis, there exists terms $w_1$ and $w_2$ of $P$ such that:

\begin{center}
$ |w_1| \equiv x~t_1~...t_{n-1}  ~$ and $~ \Gamma \vdash w_1 :  \Pi y:B^{\sharp}~C^{\sharp} $

$ |w_2| \equiv t_{n}  ~$and $~ \Gamma \vdash w_2 :  B^{\sharp} $
\end{center}

Let $u=w_1~w_2$, we have:
\begin{center}
$ |u| = |w_1|~|w_2| \equiv x~t_1~...~t_{n-1}~t_{n}~$ and
$~ \Gamma \vdash u : (w_2/y)C^{\sharp}$.
\end{center}

However, by $(\theta_3)$, Proposition \ref{equiv} and the fact that $C^{*} \equiv C^{\sharp}$, we have: 
$~A \equiv (t_n^{*}/y)C^{*} \equiv (w_2/y)C^{*} \equiv (w_2/y)C^{\sharp}~$, and, finally,
$~ \Gamma \vdash u : A~$.

\end{itemize}
\end{itemize}
\end{proof}

Finally, we get rid of the weak $\eta$-long form restriction with the
following Propositions.

\begin{proposition}
\label{etared}

 For all terms $A$, $B$ of $\lambda \Pi_P$, and for all well typed term or sort $C$ of $P$,
\begin{enumerate}
\item If $A \lra B$ then $A" \lra^{*} B"$
\item If $A \equiv B$ then $A" \equiv B"$
\item If $A$ is in weak $\eta$-long form, then $A" \lra_{\beta}^{*} A$, in particular $A" \equiv A$
\item $ \|C\|" \equiv \|C\| $
\item If $A \equiv \|C\|$ then $A" \equiv A$
\end{enumerate}
\end{proposition}

\begin{proof}
\begin{enumerate}
\item If $A \lra_{\beta} B$, then $A" \lra_{\beta} B"$  (by induction on $A$).

 If $A \lra_{\cal{R}} B$,

\begin{itemize}
\item 
for all axiom $\langle s_1,s_2 \rangle$, $(\varepsilon_{s_2} ~(\dot{s_1}))" = \varepsilon_{s_2} ~(\dot{s_1}) \lra_{\cal{R}} U_{s_1} = (U_{s_1})"$.

\item for all rule $\langle s_1,s_2,s_3 \rangle$, 
 $(\varepsilon_{s_3} (\dot{\Pi}_{\langle s_1, s_2, s_3 \rangle} C~D))" =
\\ \varepsilon_{s_3} ((\lambda x:U_{s_1} \lambda y:((\varepsilon_{s_1}~x) \Rightarrow U_{s_2})~(\dot{\Pi}_{\langle s_1, s_2, s_3 \rangle}~x~y)) C"~D")
\\ \lra_{\beta}^{2} \varepsilon_{s_3} (\dot{\Pi}_{\langle s_1,s_2, s_3 \rangle} C"~D")
 \lra_{\cal{R}}   \Pi x:(\varepsilon_{s_1}~C")~(\varepsilon_{s_2}~(D"~x)) 
\\ =   \Pi x:(\varepsilon_{s_1}~C")~(\varepsilon_{s_2}~(D~x)") $
\end{itemize}

\item By induction on the number of derivations and expansions from
$A$ to $B$. 

\item  By induction on $A$, remarking that $(\dot{\Pi}_{\langle s_1, s_2, s_3 \rangle} ~t_1~t_2)" \lra_{\beta}^{2} \dot{\Pi}_{\langle s_1, s_2, s_3 \rangle} ~t_1"~t_2"$.

\item A translated term $\|C\|$ is in weak $\eta$-long form.

\item If $A \equiv \|C\|$ then $A" \equiv \|C\|" \equiv \|C\|$, by the
the second and fourth points.
\end{enumerate}
\end{proof}

\begin{proposition}
\label{epilogue}
Let $t$ be a normal term of $\lambda \Pi_P$, 
\begin{center}
if $\|\Gamma\| \vdash t : \|A\|$ then $\|\Gamma\| \vdash t" : \|A\|$ 
\end{center}
 \end{proposition}

\begin{proof} 
 By induction on $t$.
\begin{itemize}

\item[$\bullet$] If $t$ is a well-typed product or sort, then it cannot be typed by a translated type (by confluence of $\lambda \Pi_P$). 

\item[$\bullet$] If $t=\lambda x:B~u$, 
then there exists a term $C$ of $\lambda \Pi_P$, such that \\$\|A\| \equiv \Pi x:B~C$ ($\alpha$), with $\Gamma, x:B \vdash u : C$.
\\By ($\alpha$), we have $B \equiv \|B^{*}\|$ $(\beta)$ and $C \equiv \|C^{*}\|$. Thus $\Gamma, x:\|B^{*}\| \vdash u : \|C^{*}\|$.
Then, by induction hypothesis, we have $\Gamma, x:\|B^{*}\| \vdash u" : \|C^{*}\|$,
therefore $\Gamma \vdash \lambda x:\|B^{*}\|~u" : \Pi x:\|B^{*}\|~\|C^{*}\| \equiv \|A\|$ thus $\Gamma \vdash \lambda x:B~u" :  \|A\|$, by $(\beta)$. 
Finally, by $(\beta)$ and the Proposition \ref{etared}.5,  $\lambda x:B~u" \equiv  \lambda x:B"~u"$, therefore, by subject reduction, $\Gamma \vdash t" = \lambda x:B"~u" :  \|A\|$

\item[$\bullet$] If $t$ is an application or a variable,
as it is normal, it has the form $~x~t_1~...~t_n~$ for some variable
$x$ and terms $t_1$, ..., $t_n$.  We have  $\|\Gamma\| \vdash x~t_1~...~t_n : \|A\| ~~~(\alpha_0)$.

\begin{itemize}

\item[$\multimap$] If $x$ is a variable of the context $\Sigma_P$,

\begin{itemize}

\item If $~x=\dot{s_1}~$ (where $\langle s_1,s_2\rangle$ is an axiom of $P$), 
\\then $n=0$ (because $t$ is well typed) and we have $(\dot{s_1})" = \dot{s_1}$.

\item If $~x=U_s~$ (where $s$ is a sort of $P$), then $n=0$ and $\|A\|
\equiv Type$. That's an absurdity by confluence of $\lambda
\Pi_P$.

\item If $~x=\varepsilon_s~$ (where $s$ is a sort of $P$), then, as
$t$ is well typed $n \leq 1$.

\begin{itemize}

\item[$\star$] If $n=1$, then $\|\Gamma\| \vdash t_1 : U_s$, and
$\|A\| \equiv Type$ (absurdity). 

\item[$\star$] If $n=0$, we have $(\varepsilon_s)"=\varepsilon_s$
\end{itemize}

\item If $~x=\dot{\Pi}_{\langle s_1, s_2, s_3 \rangle}~$ (where
$\langle s_1, s_2, s_3 \rangle$ is a rule of $P$), then as $t$ is
well-typed, $n \leq 2$. Moreover, $\dot{\Pi}_{\langle s_1,s_2,s_3 \rangle}$, $(\dot{\Pi}_{\langle s_1,s_2,s_3 \rangle}~t_1)$, and $(\dot{\Pi}_{\langle s_1,s_2,s_3 \rangle}~t_1~t_2)$ have the same types than their weak $\eta$-long forms.
\end{itemize}

\item[$\multimap$] If $x$ is a variable of the context $\Gamma$,

\begin{itemize}

\item If $n=0$, we have $x"=x$.
\item If $n>0$, then there exists terms $B$ and $C$ of
$\lambda \Pi_P$ such that  \\ $~\|\Gamma \| \vdash t_{n} : B
~~(\alpha_1)$ and  $~ \|\Gamma \| \vdash x~t_1~...~t_{n-1} : \Pi y:B~C
~~  (\alpha_2)~$ 
with $~\|A\| \equiv (t_{n}/y)C  ~~  (\alpha_3)~$. As in the proof of Proposition \ref{conservativity}, we can type $ x~t_1~...~t_{n-1}$ by a translated type, then  $~ \Pi y : B ~C \equiv \Pi y : \|B^{*}\| ~\|C^{*}\|~$.
In particular,
$~ B \equiv \|B^{*}\| ~$ and $~C\equiv\|C^{*}\|$. 
\\Thus, $~\|\Gamma \| \vdash t_{n} : \|B^{*}\|    ~$ and $~      \|\Gamma \| \vdash x~t_1~...~t_{n-1} : \| \Pi y : B^{*} ~C^{*}\|  $.
\\ By induction hypothesis, we have $~\|\Gamma \| \vdash t_{n}" : \|B^{*}\|    ~$ and 
\\$~      \|\Gamma \| \vdash x~t_1"~...~t_{n-1}" : \Pi y : \| B^{*} \| ~\| C^{*}\|  $.
Finally, by $(\alpha_3)$ and Proposition \ref{etared}.5,
$~\|\Gamma\|  \vdash t" =  x~t_1"~...~t_{n}" :  (t_{n}"/y)C
\equiv \|A\|$.
\end{itemize}
\end{itemize}
\end{itemize}
\end{proof}

\medskip

\begin{theorem}
Let $P$ be a functional Pure Type System, such that $\lambda \Pi_P$ is terminating.
The type $\|A\|$ is inhabited by a closed term in $\lambda \Pi_P$
if and only if the type $A$ is inhabited by a closed term in $P$.
\end{theorem}

\begin{proof}
If $A$ has a closed inhabitant in $P$, then by Proposition
\ref{correction}, $\|A\|$ has a closed inhabitant in $\lambda
\Pi_{P}$. Conversely, if $\|A\|$ has a closed inhabitant then, by
termination of $\lambda \Pi_P$ and Proposition \ref{epilogue}, it has a closed inhabitant in weak
$\eta$-long normal form and by Proposition \ref{conservativity}, $A$ has a closed inhabitant in $P$.
\end{proof}

\begin{remark}
\label{chdb}
This conservativity property we have proved is similar to that of the
Curry-de Bruijn-Howard correspondence. If the type $A^{\circ}$ is
inhabited in $\lambda \Pi$-calculus, then the proposition $A$ is
provable in minimal predicate logic, but not all terms of type
$A^{\circ}$ correspond to proofs of $A$. For instance, if $A$ is the
proposition $(\fa x~P(x)) \Ra P(c)$, then the normal term $\lambda
\alpha : (\Pi x:\iota~(P~x))~(\alpha~c)$ corresponds to a proof of $A$
but the term $\lambda \alpha : (\Pi x:\iota~(P~x))~(\alpha~((\lambda
y:\iota~y)~c))$ does not.
\end{remark}

\begin{remark}
There are two ways to express proofs of simple type theory in the
$\lambda \Pi$-calculus modulo. We can either use directly the fact
that simple type theory can be expressed in Deduction modulo
\cite{DHKHOL} or first express the proofs of simple type theory in the
Calculus of Constructions and then embed the Calculus of Constructions
in the $\lambda \Pi$-calculus modulo. 

These two solutions have some similarities, in particular if we write
$o$ the symbol $U_{Type}$. But they have also some differences: the
function $\lambda x~x$ of simple type theory is translated as the
symbol $I$ ---~or as the term $\lambda 1$~--- in the first case, using
a symbol $I$ ---~or the symbols $\lambda$ and $1$~--- specially
introduced in the context to express this particular theory, while it
is expressed as $\lambda x~x$ using the symbol $\lambda$ of the
$\lambda \Pi$-calculus modulo in the second.

More generally in the second case, we exploit the similarities of
the $\lambda \Pi$-calculus modulo and simple type theory ---~the fact that
they both allow to express functions~--- to simplify the expression
while the first method is completely generic and uses no particularity
of simple type theory. This explains why this first expression requires
only the $\lambda \Pi^-$-calculus modulo, while the second requires
the conversion rule to contain $\beta$-conversion.
\end{remark}

\end{document}